%% file: main.tex
\documentclass[twoside]{article}
\usepackage[a4paper]{geometry}
\usepackage[latin1]{inputenc} 
\usepackage[T1]{fontenc} 
\usepackage{RR}
\usepackage{hyperref}
\RRNo{9262}
\RRdate{March 2019}
\RRauthor{
Ra\'ul Pardo
  \and
Daniel Le M\'etayer}
\authorhead{R. Pardo \& D. Le M\'etayer}
\RRtitle{Analyse de politiques de vie privée pour un consentement mieux inform\'e}
\RRetitle{Analysis of Privacy Policies to Enhance Informed Consent (Extended Version)}
\titlehead{Analysis of Privacy Policies to Enhance Informed Consent (Extended Version)}
\RRresume{
Dans ce rapport, nous d\'ecrivons une m\'ethode d'analyse de politiques de vie priv\'ee qui a pour but d'am\'eliorer la qualit\'e du consentement des sujets. La d\'emarche repose sur l'utilisation d'un langage d'expression de politiques de vie priv\'ee qui peut \^{e}tre utilis\'e pour comparer et analyser les politiques. Nous d\'ecrivons un outil qui produit automatiquement les risques associ\'es \`a une politique donn\'ee pour permettre aux sujets de prendre des d\'ecisions mieux inform\'ees. L'analyse de risques est illustr\'ee par un exemple de l'internet des objets.
}

\RRabstract{
  In this report, we present an approach to enhance informed consent
  for the processing of personal data.
  The approach relies on a privacy policy language used to express,
  compare and analyze privacy policies.
  We describe a tool that automatically reports the privacy risks
  associated with a given privacy policy in order to enhance data
  subjects' awareness and to allow them to make more informed choices.
  The risk analysis of privacy policies is illustrated with an IoT
  example.
}
\RRmotcle{Langage d'expression de politiques de vie priv\'e, Risques pour la vie privée, Consentement inform\'e}
\RRkeyword{Privacy Policy Language, Privacy Risks, Informed Consent}
\RRprojets{Privatics}
\URRhoneAlpes 
\RCGrenoble 

\usepackage{amsmath}

\let\vec\relax
\DeclareMathAccent{\vec}{\mathord}{letters}{"7E} 
\usepackage{accents}

\usepackage{tikz}
\usetikzlibrary{positioning,shapes,shadows,arrows,backgrounds}
\usepackage{pgf} 
\usepackage{pgfplots} 

\usepackage{esvect}

\usepackage{paralist}

\usepackage[shadow]{todonotes}

\usepackage{amssymb}
\usepackage{stmaryrd} 

\usepackage{cancel}

\usepackage{amsthm}

\newtheorem{lemma}{Lemma}

\newtheorem{definition}{Definition}
\newtheorem{example}{Example}

\usepackage{hyperref}
\usepackage{xspace}

\usepackage{graphicx}

\usepackage{mathpartir}

\usepackage{tikz}
\usetikzlibrary{calc}

\usepackage{msc}

\usepackage{appendix}

\usepackage{tabularx}

\usepackage{xcolor,colortbl}

\usepackage{multirow}

\usepackage{times}

\input{macros}

\let\temp\phi
\let\phi\varphi
\let\varphi\temp

\begin{document}
\makeRR   

\input{introduction}

\input{language}
\input{policy-understandability}
\input{benefits-gdpr}


\input{related-work}

\input{conclusion}

\bibliographystyle{splncs04}
\bibliography{references}

\appendix
\input{appendix}

\end{document}

%% file: macros.tex
\newcommand{\Lyon}{\ensuremath{\mathit{Lyon}}\xspace}

\newcommand{\valueplatealice}{\ensuremath{\texttt{GD-042-PR}}\xspace}

\newcommand{\dataprivacydaynum}{\ensuremath{\mathit{26/04/2019}}\xspace}
\newcommand{\dataprivacydaynumnonit}{\ensuremath{26/04/2019}\xspace}
\newcommand{\gdprdaynum}{\ensuremath{\mathit{21/03/2019}}\xspace}
\newcommand{\gdprdaynumnonit}{\ensuremath{21/03/2019}\xspace}

\newcommand{\carlocation}{\ensuremath{\mathit{car{\mkern-3mu}\_location}}\xspace}
\newcommand{\watchlocation}{\ensuremath{\mathit{smartwatch{\mkern-3mu}\_location}}\xspace}

\newcommand{\numplatealice}{\ensuremath{\mathit{plate}_{\Alice}}\xspace}
\newcommand{\numplate}{\ensuremath{\mathit{number{\mkern-3mu}\_plate}}\xspace}
\newcommand{\age}{\ensuremath{\mathit{age}}\xspace}

\newcommand{\address}{\ensuremath{\mathit{address}}\xspace}
\newcommand{\city}{\ensuremath{\mathit{city}}\xspace}

\newcommand{\datatype}{\ensuremath{\mathit{t}}\xspace}
\newcommand{\dev}{\ensuremath{\mathit{dev}}\xspace}
\newcommand{\sndr}{\ensuremath{\mathit{sndr}}\xspace}
\newcommand{\rcv}{\ensuremath{\mathit{rcv}}\xspace}
\newcommand{\st}{\ensuremath{\mathit{st}}\xspace}

\newcommand{\research}{\ensuremath{\mathit{research}}\xspace}
\newcommand{\advertisement}{\ensuremath{\mathit{advertisement}}\xspace}
\newcommand{\newsletter}{\ensuremath{\mathit{newsletter}}\xspace}

\newcommand{\profiling}{\ensuremath{\mathit{profiling}}\xspace}
\newcommand{\commercialoffers}{\ensuremath{\mathit{commercial\_offers}}\xspace}

\newcommand{\rt}{\ensuremath{\mathit{rt}}\xspace}
\newcommand{\dur}{\ensuremath{\mathit{dur}}\xspace}
\newcommand{\dcr}{\ensuremath{\mathit{dcr}}\xspace}
\newcommand{\TR}{\ensuremath{\mathit{TR}}\xspace}
\newcommand{\tr}{\ensuremath{\mathit{tr}}\xspace}

\newcommand{\Google}{\ensuremath{\mathit{Google}}\xspace}
\newcommand{\Alphabet}{\ensuremath{\mathit{Alphabet}}\xspace}

\newcommand{\Parket}{\ensuremath{\mathit{Parket}}\xspace}
\newcommand{\ParketWW}{\ensuremath{\mathit{ParketWW}}\xspace}
\newcommand{\CarInsure}{\ensuremath{\mathit{CarInsure}}\xspace}
\newcommand{\AdsCom}{\ensuremath{\mathit{AdsCom}}\xspace}

\newcommand{\Alice}{\ensuremath{\mathit{Alice}}\xspace}

\newcommand{\Values}[1]{\ensuremath{\mathcal{V}_{#1}}\xspace}
\newcommand{\DataTypes}{\ensuremath{\mathcal{T}}\xspace}
\newcommand{\DataItems}{\ensuremath{\mathcal{I}}\xspace}
\newcommand{\Purposes}{\ensuremath{\mathcal{P}}\xspace}
\newcommand{\Functions}{\ensuremath{\mathcal{F}}\xspace}

\newcommand{\DataControllers}{\ensuremath{\mathcal{E}}\xspace}

\newcommand{\Entities}{\DataControllers}

\newcommand{\nat}{\ensuremath{\mathbb{N}}\xspace}

\newcommand{\wfc}{\ensuremath{\mathcal{C}}\xspace}

\newcommand{\wfur}{\ensuremath{\mathcal{D{\mkern-0mu}U{\mkern-1mu}R}}\xspace}
\newcommand{\wfcr}{\ensuremath{\mathcal{D{\mkern-0mu}C{\mkern-0mu}R}}\xspace}

\newcommand{\wfpp}[1]{\ensuremath{\mathcal{PP}_{#1}}\xspace} 

\newcommand{\Devices}{\ensuremath{\mathcal{D}}\xspace}
\newcommand{\Events}{\ensuremath{E}\xspace}

\newcommand{\PolicySet}[1]{\ensuremath{\policies_{#1}}\xspace} 
\newcommand{\ReceivedDataSet}[1]{\ensuremath{\receiveddata_{#1}}\xspace} 

\newcommand{\request}{\ensuremath{\mathit{request}}\xspace}
\newcommand{\send}{\ensuremath{\mathit{send}}\xspace}
\newcommand{\transfer}{\ensuremath{\mathit{transfer}}\xspace}

\newcommand{\use}{\ensuremath{\mathit{use}}\xspace}
\newcommand{\illegaluse}{\ensuremath{\mathit{illegal\_use}}\xspace}
\newcommand{\illegaltransfer}{\ensuremath{\mathit{illegal\_transfer}}\xspace}

\newcommand{\pur}{\ensuremath{\mathit{pur}}\xspace}

\newcommand{\policies}{\ensuremath{\pi}\xspace}

\newcommand{\dataval}{\ensuremath{\nu}\xspace}

\newcommand{\receiveddata}{\ensuremath{\rho}\xspace}

\newcommand{\poentities}{\ensuremath{\leq_{\Entities}}}
\newcommand{\popurposes}{\ensuremath{\leq_{\Purposes}}}
\newcommand{\podata}{\ensuremath{\leq_{\DataTypes}}}

\newcommand{\dursubsumes}{\ensuremath{\preceq_{\wfur}}}
\newcommand{\dcrsubsumes}{\ensuremath{\preceq_{\wfcr}}}

\newcommand{\subsumes}{\ensuremath{\sqsubseteq}}
\newcommand{\der}{\ensuremath{\vdash}}

\newcommand{\eval}[3]{\ensuremath{\texttt{eval}(#1,#2,#3)}\xspace}
\newcommand{\type}{\texttt{type}\xspace}
\newcommand{\activePolicy}{\texttt{activePolicy}\xspace}
\newcommand{\activeTransfer}{\texttt{activeTransfer}\xspace}

\newcommand{\timestamp}{\texttt{time}\xspace}

\newcommand{\owner}{\texttt{owner}\xspace}

\newcommand{\devtoent}{\ensuremath{\texttt{entity}}\xspace}

\newcommand{\mini}{\ensuremath{\mathit{min}}\xspace}

\newcommand{\true}{\ensuremath{\mathit{tt}}\xspace}
\newcommand{\false}{\ensuremath{\mathit{ff}}\xspace}

\newcommand{\policyjoin}{\ensuremath{\sqcup}}
\newcommand{\durjoin}{\ensuremath{\sqcup}_{\wfur}}
\newcommand{\dcrjoin}{\ensuremath{\sqcup}_{\wfcr}}

\newcommand{\partialfunction}{\ensuremath{\rightharpoonup}}

\newcommand*{\palicetrans}{\ensuremath{\mathit{p\_trans}_\Alice}\xspace}
\newcommand*{\palicenotrans}{\ensuremath{\mathit{p\_no\_trans}_\Alice}\xspace}

\definecolor{new-green}{RGB}{0,102,0} 

\newcommand{\dc}{\text{DC}\xspace}
\newcommand{\dcs}{\text{DCs}\xspace}
\newcommand{\ds}{\text{DS}\xspace}
\newcommand{\dss}{\text{DSs}\xspace}
\newcommand{\dcvar}{\ensuremath{\mathit{D{\mkern-3mu}C}}\xspace}

\newcommand{\etal}{\textit{et al.}\xspace}

\newcommand{\pilot}{{\sc Pilot}\xspace} 
\newcommand{\spin}{{SPIN}\xspace} 
\newcommand{\promela}{{\sc Promela}\xspace} 


%% file: introduction.tex
\section{Introduction}

One of the most common argument to legitimize the collection of personal data is the fact that the persons concerned have provided their consent or have the possibility to object to the collection. Whether opt-out is considered as an acceptable form of consent (as in the recent California Consumer Privacy Act\footnote{\url{https://leginfo.legislature.ca.gov/faces/billTextClient.xhtml?bill_id=201720180AB375}}) or opt-in is required (as in the European General Data Protection Regulation  - GDPR\footnote{\url{https://eur-lex.europa.eu/legal-content/EN/TXT/?uri=uriserv:OJ.L_.2016.119.01.0001.01.ENG&toc=OJ:L:2016:119:TOC}}), a number of conditions have to be met to ensure that the collection respects the true will of the person concerned. In fact, one may argue that this is seldom the case.
In practice, internet users generally have to consent on the fly, when they want to use a service, which leads them to accept mechanically the conditions of the provider. Therefore, their consent is generally not really informed because they do not read the privacy policies of the service providers. In addition, these policies are often vague and ambiguous. This situation, which is already critical, will become even worse with the advent of the internet of things (``IoT'') which has the potential to  extend to the ``real world'' the tracking already in place on the internet.

A way forward to address this issue is to allow users to define their own privacy policies, with the time needed to reflect on them, possibly even with the help of experts or pairs. These policies could then be applied automatically to decide upon the disclosure of their personal data and the precise conditions of such disclosures. The main benefit of this approach is to reduce the imbalance of powers between individuals and the organizations collecting their personal data (hereafter, respectively data subjects, or \dss, and data controllers, or \dcs, following the GDPR terminology): each party can define her own policy and these policies can then be compared to decide whether a given \dc is authorized to collect the personal data of a \ds. In practice, \dss can obviously not foresee all possibilities when they define their initial policies and they should have the opportunity to update them when they face new types of \dcs or new types of purposes for example. Nevertheless, their privacy policies should be able to cope with most situations and, as time passes, their coverage would become ever larger.

However, a language to define privacy policies must meet a number of requirements to be able to express the  consent of the \dss to the processing of their personal data. For example, under the GDPR, valid consent must be freely given, specific, informed and unambiguous. Therefore, the language must be endowed with a formal semantics in order to avoid any ambiguity about the meaning of a privacy policy. However, the mere existence of a semantics does not imply that \dss properly understand the meaning of a policy and its potential consequences.
One way to enhance the understanding of the \dss is to provide them information about the potential risks related to a privacy policy. This is in line with Recital 39 of the GDPR which stipulates that data subjects should be ``made aware of the risks, rules, safeguards and rights in relation to the processing of personal data and how to exercise their rights in relation to such processing''.
This approach can enhance the awareness of the \dss and allow them to adjust their privacy policies in a better informed way.

A number of languages and frameworks have been proposed in the literature to express privacy policies. However, as discussed in Section 6, none of them meets all the above requirements, especially the strong conditions for valid consent laid down by the GDPR. In this report, we define a language, called  \pilot, meeting these requirements and show its benefits to define precise privacy policies and to highlight the associated privacy risks. Even though \pilot is not restricted to the IoT, the design of the language takes into account the results of previous studies about the expectations and privacy preferences of IoT users \cite{PEPIoTWebhdbfs17}. 

We introduce the language in Section~\ref{sec:language} and its abstract execution model in Section~\ref{sec:model}. Then we show in Section~\ref{sec:risk-analysis} how it can be used to help \dss defining their own privacy policies and understanding the associated privacy risks. Because the language relies on a well-defined execution model, it is possible to reason about privacy risks and to produce (and prove) automatically answers to  questions raised by the \dss.  
In Section \ref{sec:benefits-gdpr}, we discuss the benefits of \pilot
in the context of the GDPR and also the aspects of the GDPR that
cannot be covered by a privacy policy language.
In Section~\ref{sec:related-work}, we compare \pilot with existing
privacy policy languages, and we conclude the report with avenues for
further research in Section~\ref{sec:conclusion}.




%% file: language.tex
\section{The Privacy Policy Language \pilot}
\label{sec:language}

In this section we introduce, \pilot, a privacy policy language
meeting the objectives set forth in Section 1.
The language is designed so that it can be used both by \dcs (to
define certain aspects of their privacy rules or general terms
regarding data protection) and \dss (to express their consent).
\dcs can also keep \dss policies for later use for
\emph{accountability} purposes---i.e., to show that data has been
treated in accordance with the choices of \dss.

\dcs devices must declare their privacy policies before they collect
personal data.  We refer to these policies as \emph{\dc policies}.
Likewise, when a \ds device sends data to a \dc device, the \ds device
must always include a policy defining the restrictions imposed by the
\ds on the use of her data by the \dc.
We refer to these policies as \emph{\ds policies}.

In what follows we formally define the language \pilot.
We start with definitions of the most basic elements of \pilot
(Section~\ref{subsec:basic-definitions}), which are later used to
define the abstract syntax of the language
(Section~\ref{subsec:pilot-syntax}). This syntax is then illustrated
with a working example (Section~\ref{subsec:example}).

\input{basic-definitions}

\input{syntax}

\section{Abstract execution model}
\label{sec:model}

In this section, we describe the abstract execution model of
\pilot. The purpose of this abstract model is twofold: it is useful to
define a precise semantics of the language and therefore to avoid any
ambiguity about the meaning of privacy policies; also, it is used by
the verification tool described in the next section to highlight
privacy risks.
The definition of the full semantics of the language, which is
presented in a companion paper \cite{pilot-fm-2019}, is beyond the
scope of this report. In the following, we focus on the two main
components of the abstract model: the system state (Section
\ref{subsec:system-state}) and the events (Section
\ref{subsec:semantics}).

\input{state}

\input{no-formal-semantics}

%% file: basic-definitions.tex
\subsection{Basic definitions}
\label{subsec:basic-definitions}

%
\paragraph{Devices and Entities.}
We start with a set \Devices of  \emph{devices}.
Concretely, we consider devices that are able to store, process and
communicate data.
For example, a smartphone, a laptop, an access point, an autonomous
car, etc.

Let \Entities denote the set of \emph{entities} such as Google or
Alphabet and $\poentities$ the associated partial order---e.g., since
Google belongs to Alphabet we have $\Google \poentities \Alphabet$.
Entities include \dcs and \dss.  Every device is associated with an
entity.
However, entities may have many devices associated with them.
The function $\devtoent : \Devices \to \Entities $ defines the entity
associated with a given device.

\paragraph{Data Items, datatypes and values.}
Let \DataItems be a set of \emph{data items}.
Data items correspond to the pieces of information that devices
communicate.
Each data item has a \emph{datatype} associated with it.
Let \DataTypes be a set of datatypes and  $\podata$ the associated partial order. We use function
$\type : \DataItems \to \DataTypes$ to define the datatype of each
data item.
Examples of datatypes\footnote{Note that
here we do not use the term ``datatype'' as traditionally in programming languages.
  We use datatype to refer to the semantic meaning of data items.} are: age, address, city and clinical records.
Since \city is one of the elements that the datatype \address may
be composed of, we have $\city \podata \address$.
We use $\Values{}$ to the denote the set of all values of data items, $\Values{} = (\bigcup_{t
\in \DataTypes} \Values{t})$ where $\Values{t}$ is the set of values for data items of type $t$.
We use a special element $\bot \in \Values{}$ to denote the undefined value.
A data item may be undefined, for instance, if it has been deleted or it has not been
collected.
The device where a data item is created (its source) is called the \emph{owner} device of the
data item.
We use a function $\owner \colon \DataItems \to \Devices$ to denote the owner
device of a given data item.

\paragraph{Purposes.}
We denote by \Purposes the set of \emph{purposes} and $\popurposes$ the associated partial order. 
For instance, if newsletter is considered as a specific type of advertisement, then we have  $\newsletter \popurposes \advertisement$.

\paragraph{Conditions.}
Privacy policies are contextual: they may depend  on \emph{conditions} on the information stored on the devices on which they are evaluated.
For example,
\begin{inparaenum}[(1)]
  \item\label{p1} \emph{``Only data from adults may be collected''} or
  \item\label{p2} \emph{``Only locations within the city of Lyon may
      be collected from my smartwatch''}
\end{inparaenum}
are examples of policy conditions.
%
%
In order to express conditions we use a simple logical language.
Let \Functions denote a set of functions and \emph{terms} $t$ be defined as follows: $t ::= i \mid c \mid f(\vv{t}) $ where $i
\in \DataItems$ is data item, $c \in \Values{}$ is
a constant value, $f \in \Functions$ is a function, and $\vv{t}$ is a list of
terms matching the arity of $f$.
The syntax of the logical language is as follows:
$
  \phi ::= t_1 \ast t_2 \mid \neg \phi \mid \phi_1 \wedge \phi_2 \mid \true \mid \false
$
where $\ast$ is an arbitrary binary predicate, $t_1,t_2$ are terms; \true and \false represent respectively
true and false.
For instance, $\age \geq 18$ and $\watchlocation = \Lyon$ model conditions
(\ref{p1}) and (\ref{p2}), respectively.
We denote the set of well-formed conditions as $\wfc$.
In order to compare conditions, we use a relation, $\der \colon \wfc \times \wfc$.
We write $\phi_1 \der \phi_2$ to denote that $\phi_2$ is stronger than $\phi_1$.
%

%% file: syntax.tex
\subsection{Abstract syntax of \pilot privacy policies}
\label{subsec:pilot-syntax}

In this section we introduce the abstract syntax of \pilot \emph{privacy
  policies}, or, simply, \pilot \emph{policies}. We emphasize the fact that this abstract syntax is not the syntax  used to communicate with \dss or \dcs. This abstract syntax can be associated with a concrete syntax in a restricted form of natural language. We do not describe this mapping here due to space constraints, but we provide some illustrative examples in Section \ref{subsec:example} and describe a user-friendly interface to define \pilot policies in Section \ref{subsec:usability}.
The goal of \pilot policies is to express the conditions under which data can be
communicated.
We consider two different types of data communications: \emph{data
  collection} and \emph{transfers}.
Data collection corresponds to the collection by a \dc of information directly from a \ds.
A transfer is the event of sending previously collected data to third
parties.
%
%
%
%

\begin{definition}[\pilot Privacy Policies Syntax]
Given
$\mathit{Purposes} \in 2^\Purposes$,
$\mathit{retention\_time} \in \nat$,
$\mathit{condition} \in \wfc$,
$\mathit{entity} \in \Entities$ and
$\mathit{datatype} \in \DataTypes$,
the syntax of \pilot policies is defined as follows:

\begin{tabular}{rcl}
  $\pilot \, \mathit{Privacy \, Policy}$             & ::= &
                                                     $(\mathit{datatype}, \dcr, \TR)$ \\
  $\mathit{Data \, Communication \, Rule} \, (\dcr)$ & ::= &
                                                     $\langle \mathit{condition}, \mathit{entity}, \dur \rangle$\\
                                                     $\mathit{Data \, Usage \, Rule} \, (\dur)$         & ::= &
                                                     $\langle \mathit{Purposes}, \mathit{retention\_time} \rangle$ \\
  $\mathit{Transfer \, Rules}$ $(\TR)$               & ::= &
                                                     $\{\dcr_1, \dcr_2, \ldots\}$
\end{tabular}
\end{definition}

We use \wfur, \wfcr, \wfpp{} to denote the sets of data usage rules, data
communication rules and \pilot privacy policies, respectively.
The set of transfer rules is defined as the set of sets of data communication rules, $\TR
\in 2^\wfcr$.
In what follows, we provide some intuition about this syntax and an example of application.
\begin{compactdesc}
\item[Data Usage Rules.] The purpose of these rules is to define the
operations that may be performed on the data.
  $\mathit{Purposes}$ is the set of allowed purposes and
  $\mathit{retention\_time}$ the deadline for erasing the data.
  %
  %
  As an example, consider the following data usage rule,
  $$\dur_1 = \langle \{\research\}, \dataprivacydaynum \rangle.$$
  This rule states that the  data may be used only for the purpose of
  research and may be used until \dataprivacydaynum.

\item[Data Communication Rules.]
  A data communication rule defines the conditions that must be met for the data to be
  collected by or communicated to an entity.
  The outer layer of data communication rules --- i.e., the condition and entity --- should be checked by the sender whereas the data usage rule is to be enforced by the receiver.
  The first element, $\mathit{condition}$, imposes constraints on the data item and the context (state of the \ds device);
  $\mathit{entity}$ indicates the entity allowed to receive the data;
  $\dur$ is a data usage rule stating how $\mathit{entity}$ may use the data.
  For example,
  $ \langle \age > 18, \AdsCom, \dur_1 \rangle. $
  states that data may be communicated to the entity AdsCom which may use it according to $\dur_1$ (defined above).
  It also requires that the data item $\age$ is greater than $18$.
  This data item may be the data item to be sent or part of the contextual information of the sender device.

\item[Transfer Rules.]
  These rules form a set of data communication rules specifying
  the entities to whom the data may be transferred.

\item[\pilot Privacy Policies.]
  \dss and \dcs use \pilot policies to describe how data may be used, collected and transferred.
  The first element, $\mathit{datatype}$, indicates the type of data the policy
  applies to;
  \dcr defines the collection conditions and \TR the transfer rules.
  In some cases, several \pilot policies are necessary to fully capture the
  privacy choices for a given datatype.
  For instance, a \ds may allow only her employer
  to collect her data when she is at work but, when being in a museum, she may allow only the museum.
  In this example, the \ds must define two policies, one for each location.
\end{compactdesc}

\subsection{Example: Vehicle Tracking}
\label{subsec:example}
  In this section, we illustrate the syntax of \pilot with a concrete example that will be continued with the risk analysis in Section \ref{sec:risk-analysis}. 
  
  The use of Automatic Number Plate Recognition (ANPR)~\cite{DSBalprasar13} is
  becoming very popular for applications such as parking billing or pay-per-use roads.
  These systems consist of a set of cameras that automatically recognize plate numbers when vehicles cross the range
  covered by the cameras.
  Using this information, it is possible to determine how long a car
  has been in a parking place or how many times it has traveled on a highway, for example.

  ANPR systems may collect large amounts of mobility data, which
  raises  privacy concerns~\cite{effalpr}.
  When data is collected for the purpose of billing, the consent of the customer is not needed since the legal ground for the data processing can be the performance of a contract. However, certain privacy regulations, such as the GDPR, require prior consent for the use of the data for other purposes, such as commercial offers or advertisement.

  Consider a \dc, Parket, which owns parking areas equipped with ANPR in France.
  Parket is interested in offering discounts to frequent customers.
  To this end, Parket uses the number plates recorded by the ANPR system to send commercial offers to a selection of customers.
  Additionally, Parket transfers some data to its sister company, ParketWW, that operates worldwide with the goal of providing better offers to their customers.
  Using  data for these purposes requires explicit consent from \dss.
  The \pilot policy below precisely captures the way in which Parket wants to collect and use number plates for these purposes.
  \begin{equation}\label{eq:parket-policy}
    \begin{tabular}{rl}
      $(\numplate,$ & $\langle \true, \Parket, \langle \{\commercialoffers\}, \gdprdaynum \rangle \rangle,$\\
      \,            & $\{\langle \true, \ParketWW, \langle \{\commercialoffers\}, \dataprivacydaynum \rangle \rangle\})$
    \end{tabular}      
  \end{equation}
  The condition (\true) in (\ref{eq:parket-policy}) means that Parket does not impose any condition on the number plates it collects or transfers to ParketWW.
  This policy can be mapped into the following  natural
  language sentence:
  \begin{quote}
    \Parket may collect data of type \numplate and use it for \commercialoffers purposes until $\gdprdaynum$.
    
    This data may be transferred to
    \ParketWW which may use it for \commercialoffers purposes until $\dataprivacydaynum$.
  \end{quote}
  The parts of the policy in $\mathit{italic}$ $\mathit{font}$ correspond to the elements of \pilot's abstract syntax.
  These elements change based on the content of the policy.
  The remaining parts of the policy are common to all \pilot policies.
  
  To obtain \dss consent, Parket uses a system which broadcasts the above
  \pilot policy  to  vehicles before they enter the ANPR area. The implementation of this broadcast process is outside the scope of this report; several solutions are presented in 
 \cite{lemetayer:hal-01953052}.
 \dss are therefore informed about Parket's policy before data is
  collected.
  However, \dss may disagree about the processing of their data for these purposes. They can express their own privacy policy in \pilot to define the conditions of their consent (or denial of consent).

  Consider a \ds, Alice, who often visits Parket parkings.
  Alice wants to benefit from the offers that Parket provides in her city (Lyon) but does
  not want her information to be transferred to third-parties.
  To this end, she uses the following \pilot policy:
  \begin{equation}\label{eq:alice-policy}
    \begin{tabular}{rl}
     $(\numplate,$ & $\langle \carlocation = \Lyon, \Parket,$ \\
     & $\langle \{\commercialoffers\}, \gdprdaynum \rangle \rangle, \emptyset)$
    \end{tabular}
  \end{equation}
 In practice, she would actually express this policy as follows:
  \begin{quote}
    \Parket may collect data of type \numplate if \carlocation $\mathit{is \, Lyon}$ and use it for \commercialoffers
    purposes until $\gdprdaynum$.
  \end{quote}
  which is a natural language version of the above abstract syntax policy.
  
  In contrast with Parket's policy, Alice's policy includes a condition
  using \carlocation, which is a data item containing the current location of Alice's car.
  In addition, the absence of transfer statement means that Alice does not allow
  Parket to transfer her data.
  It is easy to see that 
  Alice's policy is more restrictive than Parket's policy.
  Thus, after Alice's device\footnote{This device can be Alice's car on-board computer, which can itself be connected to the mobile phone used by Alice to manage  her privacy policies \cite{lemetayer:hal-01953052}.} receives Parket's policy, it can automatically send an answer
  to Parket indicating that Alice does not give her consent to the collection of her data in the conditions stated in  \Parket's policy. In practice, Alice's policy can also be sent back so that Parket can possibly adjust her own policy to match Alice's requirements. Parket would then have the option to send a new  \dc policy consistent with Alice's policy and Alice would send her consent in return. The new policy sent by Parket can be computed as a join of Parket's original policy and Alice's policy (see Appendix~\ref{sec:policy-join} for an example of policy join which is proven to preserve the privacy preferences of the \ds).
  
  This example is continued in Section \ref{sec:risk-analysis} which illustrates the use of \pilot to enhance Alice's awareness by providing her information about the risks related to her choices of privacy policy.

%% file: state.tex
\subsection{System state}
\label{subsec:system-state}

We first present an abstract model of a system composed of devices that
communicate information and use \pilot policies to express the privacy
requirements of \dss and \dcs.
Every device has a set of associated policies.
A policy is associated with a device if it was defined in the device
or the device received it.
Additionally, \ds devices have a set of data associated with them.
These data may represent, for instance, the MAC address of the device or
workouts  recorded by the device.
Finally, we keep track of the data collected by \dc devices together
with their corresponding \pilot policies.
The system state is formally defined as follows.

\begin{definition}[System state]\label{def:system-state}
  The system state is a triple $\langle
                                  \dataval,
                                  \policies,
                                  \receiveddata
                                \rangle$
  where:
  \begin{itemize}
  \item
    $\dataval : \Devices \times \DataItems \partialfunction \Values{}$
    is a mapping from the data items of a device to their
    corresponding value in that device.

  \item
    $\policies : \Devices \rightarrow 2^{\Devices \times \wfpp{}}$ is
    a function denoting the \emph{policy base} of a device.
    The policy base contains the policies created by the owner of the
    device and the policies sent by other devices in order to state
    their collection requirements.
    A pair $(d,p)$ means that \pilot policy $p$ belongs to device $d$.
    We write $\PolicySet{d}$ to denote $\policies(d)$.

  \item
    $\receiveddata : \Devices \rightarrow 2^{\Devices \times \DataItems \times
    \wfpp{}}$ returns a set of triples $(s,i,p)$ indicating the data items and
    \pilot policies that a controller has received.
    If $(d',i,p) \in \receiveddata(d)$, we say that device $d$ has received or
    collected data item $i$ from device $d'$ and policy $p$ describes how the
    data item must be used.
    We write $\ReceivedDataSet{d}$ to denote $\receiveddata(d)$.

  \end{itemize}

\end{definition}

In Definition~\ref{def:system-state}, \dataval returns the local value
of a data item in the specified device.
However, not all devices have values for all data items.
When the value of a data item in a device is undefined, \dataval returns $\bot$.
The policy base of a device $d$, $\policies_d$, contains the \pilot policies
that the device has received or that have been defined locally.
If $(d,p) \in \policies(d)$, the policy $p$ corresponds to a policy that $d$ has
defined in the device itself.
On the other hand, if $(d,p) \in \policies(e)$ where $d \not = e$, $p$ is a
policy sent from device $e$.
Policies stored in the policy base are used to compare the privacy policies of two
devices before the data is communicated.
The information that a device has received is recorded in \receiveddata.
Also, \receiveddata contains the \pilot policy describing how data must be used.
The difference between policies in \policies and \receiveddata is that
policies in \policies are used to determine whether data can be
communicated, and policies in \receiveddata are used to describe how a
data item must be used by the receiver.

\begin{example}\label{ex:system-state}
  %
  Fig.~\ref{fig:state-example} shows a state composed of two
  devices: Alice's car, and Parket's ANPR system.
  The figure depicts the situation  after Alice's car has
  entered the range covered by the ANPR camera and the collection of her data has already occurred.
  %

  %
  The database in Alice's state ($\dataval_\Alice$) contains a data
  item of type number plate \numplatealice whose value is
  \valueplatealice.
  The policy base in Alice's device ($\policies_\Alice$) contains two
  policies: $(\Alice,p_\Alice)$ representing a policy that Alice
  defined, and $(\Parket,$ $p_\Parket)$ which represents a policy
  $p_\Parket$ sent by Parket.
 We assume that $p_\Alice$ and $p_\Parket$ are
  the policies applying to data items of type number plate.

  %
  Parket's state contains the same components as Alice's state with, in addition, a set of received data ($\receiveddata_\Parket$).
  The latter contains the data item \numplatealice collected from Alice  and the \pilot policy $p_\Parket$ that must be applied in order to handle the data.
  %
  %
  Note that $p_\Parket$ was the \pilot policy originally defined by
  Parket.
  In order for Alice's privacy to be preserved, it must hold that
  $p_\Parket$ is more restrictive than Alice's \pilot policy
  $p_\Parket$, which is denoted by $p_\Parket \subsumes p_\Alice$.\footnote{See Appendix~\ref{sec:policy-subsumption} for the formal definition of $\subsumes$.}
  This condition can easily be enforced by comparing the policies
  before data is collected.
  The first element in $(\Alice, \numplatealice, p_\Parket)$ indicates
  that the data comes from Alice's device.
  %
  %
  %
  Finally, Parket's policy base has one policy: its own policy
  $p_\Parket$, which was communicated to Alice for data collection.\qed
  %
  %
  %
  %
\end{example}

\input{example-of-state-2}

%% file: example-of-state-2.tex

\tikzstyle{abstract}=[rectangle, draw=black, rounded corners, fill=white,
                      text centered, anchor=north, text=black, text width=3.5cm]

\begin{figure}[!t]
\centering
\scalebox{0.7}{
\begin{tikzpicture}
     \node (car)
        {
          \includegraphics{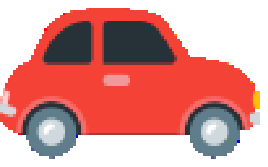}
        };

    \node (piSubject) [abstract,
                       rectangle split,
                       rectangle split parts=3,
                       left=of car,
                       yshift=7mm]
        {
          $\boldsymbol{\policies_\Alice}$
          \nodepart{two}$(\Alice,p_\Alice)$
          \nodepart{three}$(\Parket,p_\Parket)$
        };
    \node (datavalSubject) [abstract,
                             rectangle split,
                             rectangle split parts=2,
                             below=of piSubject,
                             yshift=8mm
                            ]
        {
          $\boldsymbol{\dataval_\Alice}$            
          \nodepart{two}$(\numplatealice,\valueplatealice)$         
        };
    \node (aux1)  [below=of piSubject,
                   yshift=1cm]{};
    \begin{pgfonlayer}{background}
      \filldraw [line width=4mm,join=round,black!10]
      (piSubject.north  -| piSubject.east)  rectangle (datavalSubject.south  -| datavalSubject.west);
    \end{pgfonlayer}

    \node (camera-gate) [right=of car,
                         xshift = 2cm,
                         yshift = 0.2cm]
        {
          \includegraphics[width=2cm]{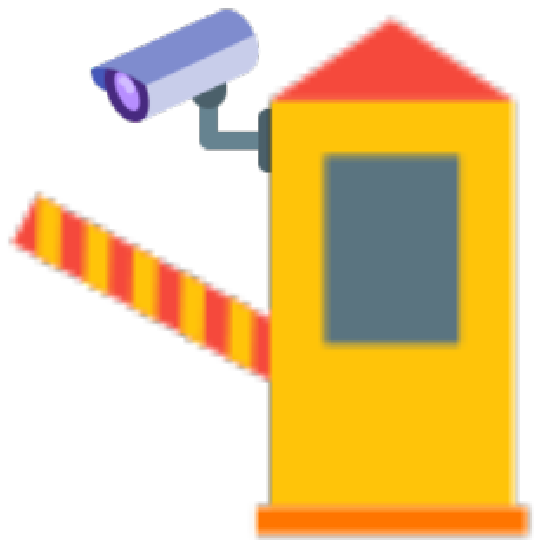}
        };

    \node (datavalController) [abstract,
                               rectangle split,
                               rectangle split parts=2,
                               right=of camera-gate]
        {
            $\boldsymbol{\dataval_\Parket}$
            \nodepart{second}$(\numplatealice,\valueplatealice)$
        };

    \node (piController) [abstract,
                          rectangle split,
                          rectangle split parts=2,
                          above = of datavalController,
                          yshift = -8mm]
        {
            $\boldsymbol{\policies_\Parket}$
            \nodepart{two}$(\Parket,p_\Parket)$
        };

    \node (rhoController) [abstract,
                           rectangle split,
                           rectangle split parts=2,
                           below=of datavalController,
                           yshift=8mm]
        {
            $\boldsymbol{\receiveddata_\Parket}$
            \nodepart{second}$(\Alice,\numplatealice,p_\Parket)$
        };
    \begin{pgfonlayer}{background}
      \filldraw [line width=4mm,join=round,black!10]
      (piController.north  -| piController.east) rectangle (rhoController.south  -| rhoController.west);
    \end{pgfonlayer}

\end{tikzpicture}
}
\caption{Example System State}
\label{fig:state-example}
\end{figure}


%% file: no-formal-semantics.tex
\subsection{System events}\label{subsec:semantics}

In this section we describe the set of events $\Events$ in our
abstract execution model.
We focus on events that ensure that the exchange of data items is done
according to the \pilot policies of \dss and \dcs.

\paragraph{Events.} The set of events $\Events$ is composed by the
following the events: $\request$, $\send$, $\transfer$ and $\use$.
%
%
The events \request, \send and \transfer model valid exchanges of
policies and data among \dcs and \dss.
The event \use models correct usage of the collected data by \dcs.
In what follow we explain each event in detail.


\begin{itemize}
\item[$\request(\sndr,\rcv,\datatype,p)$]
  models request of data from \dcs to \dss or other \dcs.
  Thus, \sndr is always a \dc device, and \rcv may be a \dc or \ds
  device.
  A request includes the type of the data that is being
  requested $\datatype$ and a \pilot policy $p$.
  As expected, the \pilot policy is required to refer to the datatype
  that is requested, i.e., $p = (\datatype,\_,\_)$.
  As a result of executing \request, the pair $(\rcv, p)$ is added to $\PolicySet{\rcv}$.
  Thus, \rcv is informed of the conditions under which \sndr will use
  the requested data.

\item[$\send(\sndr,\rcv,i)$]
  represents the collection by the \dc $\rcv$ of
  a data item $i$ from the \ds $\sndr$.
 %
  In order for \send to be executed, the device \sndr must check that $\PolicySet{\sndr}$ contains:
  \begin{inparaenum}[i)]
  \item an active policy defined by \sndr, $p_\sndr$, indicating how
    \sndr allows \dcs to use her data, and
  \item an active policy sent by \rcv, $p_\rcv$, indicating how she plans to use the data.
  \end{inparaenum}
  A policy is active if it applies to the data item to be sent, to \rcv's entity, the retention time has not yet been reached, and its condition holds.\footnote{\label{fot:trans}See Appendix~\ref{sec:active-policy-transfer} for the formal definition of active policy and active transfer.}
  Data can only be sent if $p_\rcv$ is more restrictive than $p_\sndr$ (i.e., $p_\rcv \subsumes p_\sndr$), which must be checked by \sndr. 
  We record the data exchange in $\ReceivedDataSet{\rcv}$ indicating:
  the sender, the data item and  \rcv's \pilot policy,
  $(\sndr,i,p_\rcv)$.
  We also update \rcv's database with the value of $i$ in \sndr's
  state, $\dataval(\rcv,i) = \dataval(\sndr,i)$.

\item[$\transfer(\sndr,\rcv,i)$]
  is executed when a \dc\ (\sndr) transfers a data item $i$ to
  another \dc\ (\rcv).
  First, \sndr checks whether $\PolicySet{\sndr}$ contains an active policy, from \rcv, $p_\rcv$.
  Here we do not use a \pilot policy from \sndr, instead we use the
  \pilot policy $p$ sent along with the data---defined by the owner of
  $i$.
  Thus, \sndr must check whether there exists an active transfer rule (\tr) in the set
  of transfers rules of the \pilot policy $p$.
  As before, \sndr must check that the policy sent by $\rcv$ is more
  restrictive than those originally sent by the owner of the data,
  i.e., $p_\rcv \subsumes p_\tr$ where $p_\tr$ is a policy with the
  active transfer $\tr$ in the place of the data communication rule
  and with the same set of transfers as $p$.
  Note that data items can be transferred more than once to the entities in the set of transfers as long as the retention time has not been reached.
  This is not an issue in terms of privacy as data items are constant values.
  In the resulting state, we update $\ReceivedDataSet{\rcv}$ with the
  sender, the data item and \rcv's \pilot policy,
  $(\sndr,i,p_{\rcv})$.
  Note that, in this case, the owner of the data item is not \sndr
  since transfers always correspond to exchanges of previously
  collected data, $\owner(i) \not = \sndr$.
  The database of \rcv is updated with the current value of $i$ in
  $\dataval_\sndr$.

\item[$\use(\dev,i,\pur)$] models the use of a data item $i$ by a \dc
  device $\dev$ for purpose $\pur$.
  Usage conditions are specified in the data usage rule of the policy
  attached to the data item, denoted as $p_i$, in the set of received
  data of \dev, \ReceivedDataSet{\dev}.
  Thus, in order to execute \use we require that:
  \begin{inparaenum}[i)]
  \item the purpose \pur is allowed by $p_i$, and
  \item the retention time in $p_i$ has not elapsed.
  \end{inparaenum}
\end{itemize}

%% file: policy-understandability.tex
\section{Risk Analysis}\label{sec:risk-analysis}

As described in the introduction, an effective way to enhance informed consent is to raise user awareness about the risks related to personal data collection. 
Privacy risks may result  from different sorts of misbehavior such as the use of data beyond the allowed purpose  or the transfer of  data to unauthorized third parties \cite{sjd-dlm-book}.

In order to assess the risks related to a given privacy policy, we need to rely on  assumptions about potential risk sources, such as:
\begin{itemize}
  \item Entities $e_i$ that may have a strong interest to use data of type $t$ for a given purpose $\pur$.
  \item Entities $e_i$ that may have facilities and interest to transfer data of type $t$ to other entities $e_j$.
\end{itemize}

In practice, some of these assumptions may be generic and could be obtained from databases populated by pairs or NGOs based on history of misconducts by companies or business sectors. Others risk assumptions can be specific to the DS (e.g., if she fears that a friend may be tempted to transfer certain information to another person).
Based on these assumptions, a DS who is wondering whether she should add a policy $p$ to her current set of policies can ask questions such as: ``if I add this policy $p$:

\begin{itemize}
  \item Is there a risk that my data of type $t$ is used for purpose $\pur$?
  \item Is there a risk that, at some stage, entity $e$ gets my data of type $t$? ''
\end{itemize}

In what follows, we first introduce our approach to answer the above questions (Section~\ref{subsec:arr}); then we illustrate it with the example introduced in Section~\ref{subsec:example}  (Section~\ref{subsec:case}) and we present a user-friendly interface to define and analyze privacy policies (Section~\ref{subsec:usability}).

\subsection{Automatic Risk Analysis with \spin}
\label{subsec:arr}
%
In order to automatically answer questions of the type described above,
we use the verification tool \spin~\cite{spin}.
\spin belongs to the family of verification tools known as
\emph{model-checkers}.
A model-checker takes as input a model of the system (i.e., an abstract
description of the behavior of the system) and a set of properties
(typically expressed in formal logic), and checks whether the model of
the system satisfies the properties.
In \spin, the model is written in the modeling language
\promela~\cite{spin} and properties are encoded in \emph{Linear
  Temporal Logic} (LTL)~(e.g., \cite{ltl}).
We chose \spin as it has successfully been used in a variety of contexts~\cite{MGLpaactavlpp06}.
However, our methodology is not limited to \spin and any other formal verification tool
such as SMT solvers~\cite{smt} or automated theorem provers~\cite{atp}
could be used instead.

Our approach consists in defining a \promela model for the \pilot events and
privacy policies, and translating the risk analysis questions into LTL
properties that can be automatically checked by \spin. For example, the question \emph{``Is there a risk that Alice's data is used for the purpose of profiling by ParketWW?''} is translated into the LTL property \emph{``ParketWW never uses Alice's data for profiling''}.
Devices are modeled as processes that randomly try to
execute events defined as set forth in Section~\ref{subsec:semantics}.

In order to encode the misbehavior expressed in the assumptions, we add ``illegal'' events to the set of events that devices can
execute.
For instance, consider the assumption \emph{``use of data beyond the
  allowed purpose''}.
To model this assumption, we introduce the event
 $\illegaluse$, which behaves as $\use$, but disregards
 the purpose of the \ds
policy for the data.
%

\spin explores all possible sequences of executions of
events (including misbehavior events) trying to find a sequence that violates the LTL property.
If no sequence is found, the property cannot be violated, which means that the risk corresponding to the property cannot occur.
If a sequence is found, the risk corresponding to the property can occur, and \spin returns the sequence of events that leads to the violation.
This sequence of events can be used to further clarify the cause of the violation.

\subsection{Case Study: Vehicle Tracking}
\label{subsec:case}
We illustrate our risk analysis technique with the vehicle tracking example
 introduced in Section~\ref{subsec:example}.
We first define the \promela model and the assumptions on the entities
involved in this example.
The code of the complete model is available in~\cite{promelamodel}.

\paragraph{Promela Model.}

We define a model involving the three entities identified in Section~\ref{subsec:example} with, in addition, the car insurance company $\CarInsure $ which is identified as a potential source of risk related to $\ParketWW$, i.e., 
$\Entities = \{\Alice, $ $\Parket,$ $\ParketWW,$ $\CarInsure\}$.
Each entity  is associated with a single device: 
$\Devices = \Entities$ and $\devtoent(x) = x$ for
$x \in \{\Alice, $ $\Parket,$ $\ParketWW,$ $\CarInsure\}$.
We focus on one datatype $\DataTypes = \{\numplate\}$ with its set of values
defined as $\Values{\numplate} =$ $\{\valueplatealice\}$.
We consider a data item $\numplatealice$ of type $\numplate$ for which \Alice is the owner.
Finally, we consider a set of purposes
$\Purposes = \{\commercialoffers,$ $\profiling\}$. 

\paragraph{Risk assumptions on entities.}
In this case study, we consider two risk assumptions:
\begin{enumerate}
  \item ParketWW may transfer personal data to CarInsure disregarding
  the associated \ds  privacy policies.
  \item CarInsure has strong interest in using personal data for profiling.
\end{enumerate}
In practice, these assumptions, which are not specific to Alice, may be obtained automatically from databases populated by pairs or NGOs for example.

\paragraph{Set of events.}
The set of events that we consider is derived from the risk assumptions on entities.
%
%
On the one hand, we model events that behave correctly, i.e., as described in Section~\ref{subsec:semantics}.
In order to model the worst case scenario in terms of risk analysis, we consider that: the \dcs in this case study (i.e., Parket, ParketWW and CarInsure) can request data to any entity (including Alice), the \dcs can collect Alice's data, and the \dcs can transfer data among them.
On the other hand, the risk assumptions above are modeled as two events: ParketWW may transfer data to CarInsure disregarding Alice's policy, and CarInsure may use Alice's data for profiling even if it is not allowed by Alice's
policy.
%
%
Let $\mathit{\dcvar},\mathit{\dcvar'} \in \{\Parket, \ParketWW, \CarInsure\}$, the following events may occur:
\begin{itemize}
  \item $\request(\dcvar,\Alice,\numplate,p)$ - A \dc requests a
    number plate from Alice and $p$ is the \pilot policy of the \dc.
  \item $\request(\dcvar,\dcvar',\numplate,p)$ - A \dc requests data items of type number plate from another \dc and $p$ is the \pilot policy of the requester \dc.
  \item $\send(\Alice,\dcvar,i)$ - Alice sends her item $i$ to a \dc.
  \item $\transfer(\dcvar,\dcvar',i)$ - A \dc transfers a previously received item $i$ to another \dc.
  \item $\illegaltransfer(\ParketWW,\CarInsure,i)$ - ParketWW transfers a
    previously received item $i$ to CarInsure disregarding the associated \pilot policy defined by the owner of $i$.
  \item $\illegaluse(\CarInsure,i,\profiling)$ - CarInsure uses data item
    $i$ for profiling disregarding the associated privacy policy defined by the owner of $i$.
\end{itemize}

\paragraph{Alice's policies.}
In order to illustrate the benefits of our risk analysis
approach, we focus on the following two policies that Alice may consider.
$$
\begin{array}{rl}
  \palicetrans = & (\numplate, \langle
                  \true,
                  \Parket,
                  \langle
                    \{\commercialoffers\},
                    \gdprdaynum
                  \rangle
                \rangle,\\
                &\{
                  \langle
                    \true,
                    \ParketWW,
                    \langle
                      \{\commercialoffers\},
                      \dataprivacydaynum
                    \rangle
                  \rangle\}).
\end{array}
$$
$$
\palicenotrans = (\numplate, \langle \true, \Parket, \langle \{\commercialoffers\}, \gdprdaynum \rangle \rangle, \emptyset).
$$
The policy $\palicetrans$ states that Parket can collect data of type
number plate from Alice, use it for commercial offers and keep it until
\gdprdaynumnonit.
It also allows Parket to transfer the data to ParketWW.
ParketWW may use the data for commercial offers and keep it until
\dataprivacydaynumnonit.
The policy $\palicenotrans$ is similar to $\palicetrans$ except that it does not
allow Parket to transfer the data. We assume that Alice has not yet defined any other privacy policy concerning Parket and ParketWW.

\paragraph{Parket's policy.}
We set Parket's privacy policy equal to Alice's.
By doing so, we consider the worst case scenario in terms of privacy risks because it allows Parket to collect Alice's
data and use it in all conditions and for all purposes allowed by Alice.

\paragraph{ParketWW's policy.}
%
%
%
Similarly, ParketWW's policy is aligned with the transfer rule in $\palicetrans$:
$$
\begin{array}{rl}
  p_\ParketWW = & (\numplate, \langle
                  \true,
                  \ParketWW,
                  \langle 
                   \{\commercialoffers\},
                    \dataprivacydaynum
                  \rangle
                \rangle,\emptyset).
\end{array}
$$
The above policy states that ParketWW may use data of type  number plate for commercial offers and keep it until \dataprivacydaynum.
It also represents the worst case scenario for risk analysis, as it
matches the preferences in Alice's first policy.

\subsection*{Results of the Risk Analysis}
Table~\ref{tab:risk-analysis} summarizes some of the results of the application of our \spin risk analyzer on this example.
%
%
%
%
The questions in the first column have been translated into LTL properties used by
\spin (see~\cite{promelamodel}).
%
%
The output of \spin appears in columns 2 to 5.
The green boxes indicate that the output is in accordance with Alice's policy while red boxes correspond to violations of her policy.

Columns 2 and 3 correspond to executions of the system involving correct events, considering respectively \palicetrans and \palicenotrans as Alice's policy.
As expected, all these executions respect Alice's policies.

Columns 3 and 4 consider executions involving \illegaltransfer and \illegaluse.
These columns show the privacy risks taken by  Alice based on the above risk assumptions.
%
%
%
%
Rows 3 and 6 show respectively that CarInsure may get Alice's data and use it for profiling.
In addition, the counterexamples generated by \spin, which are not pictured in the table, show that this can happen only after ParketWW executes \illegaltransfer.

From the results of this privacy risk analysis Alice may take a better informed decision about the policy to choose.
In a nutshell, she has three options:
\begin{enumerate}
\item Disallow Parket to use her data for commercial offers, i.e., choose to add neither \palicetrans nor \palicenotrans to her set of policies (Parket will use the data only for  billing purposes, based on contract).
\item Allow Parket to use her data for commercial offers without transfers to ParketWW, i.e., choose \palicenotrans.
\item Allow Parket to use her data for commercial offers and to transfer to ParketWW, i.e., choose \palicetrans.
\end{enumerate}
Therefore, if Alice wants to receive commercial offers but does not want to take the risk of being profiled by an insurance company, she should take option two.

\newcolumntype{C}[1]{>{\centering\arraybackslash}p{#1}}
\newcolumntype{M}[1]{>{\centering\arraybackslash}m{#1}}
\begin{table}[!t]
  \centering
  \begin{tabular}{|M{3.7cm}|M{1.8cm}|M{2.1cm}|M{1.8cm}|M{2.1cm}|}
    \hline
    \multirow{2}{*}{Question} & 
    \multicolumn{2}{c|}{Normal behavior} &
    \multicolumn{2}{c|}{Misbehavior Assumptions} \\
    \cline{2-5}
    & 
    $\palicetrans$ &
    $\palicenotrans$&
    $\palicetrans$ &   
    $\palicenotrans$ \\
    \hline
    Can Parket receive Alice's data? &
    \cellcolor{green!50} Yes &
    \cellcolor{green!50} Yes &
    \cellcolor{green!50} Yes &
    \cellcolor{green!50} Yes \\
    \hline
    Can ParketWW receive Alice's data? &
    \cellcolor{green!50} Yes &
    \cellcolor{green!50} No &
    \cellcolor{green!50} Yes &
    \cellcolor{green!50} No \\
    \hline
    Can CarInsure receive Alice's data? &
    \cellcolor{green!50} No &
    \cellcolor{green!50} No &
    \cellcolor{red!50} Yes &
    \cellcolor{green!50} No \\
    \hline
    Can Parket use Alice's data for other purpose than commercial offers? &
    \cellcolor{green!50} No &
    \cellcolor{green!50} No &
    \cellcolor{green!50} No &
    \cellcolor{green!50} No \\
    \hline
    Can ParketWW use Alice's data for other purpose than commercial offers? &
    \cellcolor{green!50} No &
    \cellcolor{green!50} No &
    \cellcolor{green!50} No &
    \cellcolor{green!50} No \\
    \hline
    Can CarInsure use Alice's data for profiling? &
    \cellcolor{green!50} No &
    \cellcolor{green!50} No &
    \cellcolor{red!50} Yes &
    \cellcolor{green!50} No \\
    \hline

  \end{tabular}
  \caption{Risk Analysis of Alice's policies $\palicetrans$ and
    $\palicenotrans$. Red boxes denote that Alice's policy is
    violated. Green boxes denote that Alice's policy is respected.}
  \label{tab:risk-analysis}
\end{table}

\subsection{Usability}\label{subsec:usability}

In order to show the usability of the approach, we  have developed a web application to make it possible for users with no technical background to perform risk analysis as outlined in Section~\ref{subsec:case} for the ANPR system.
The web application is available online at:  \url{http://pilot-risk-analysis.inrialpes.fr/}.

Fig.~\ref{fig:web-app-risk-analysis} shows the input forms of the web application.
First, \dss have access to a user-friendly form to input \pilot policies.
In the figure we show an example for the policy \palicetrans.
Then \dss can choose the appropriate risk assumptions from the list generated by the system.
Finally, they can ask questions about the potential risks based on these assumptions.
When clicking on ``Verify!'', the web application runs SPIN to verify the LTL property corresponding to the question.
The text ``Not Analyzed'' in grey is updated with ``Yes'' or ``No'' depending on the result.
The figure shows the results of the three first questions with \palicetrans and no risk assumption chosen (first column in Table~\ref{tab:risk-analysis}).

\begin{figure}[!t]
\centering
\includegraphics[width=0.9\textwidth]{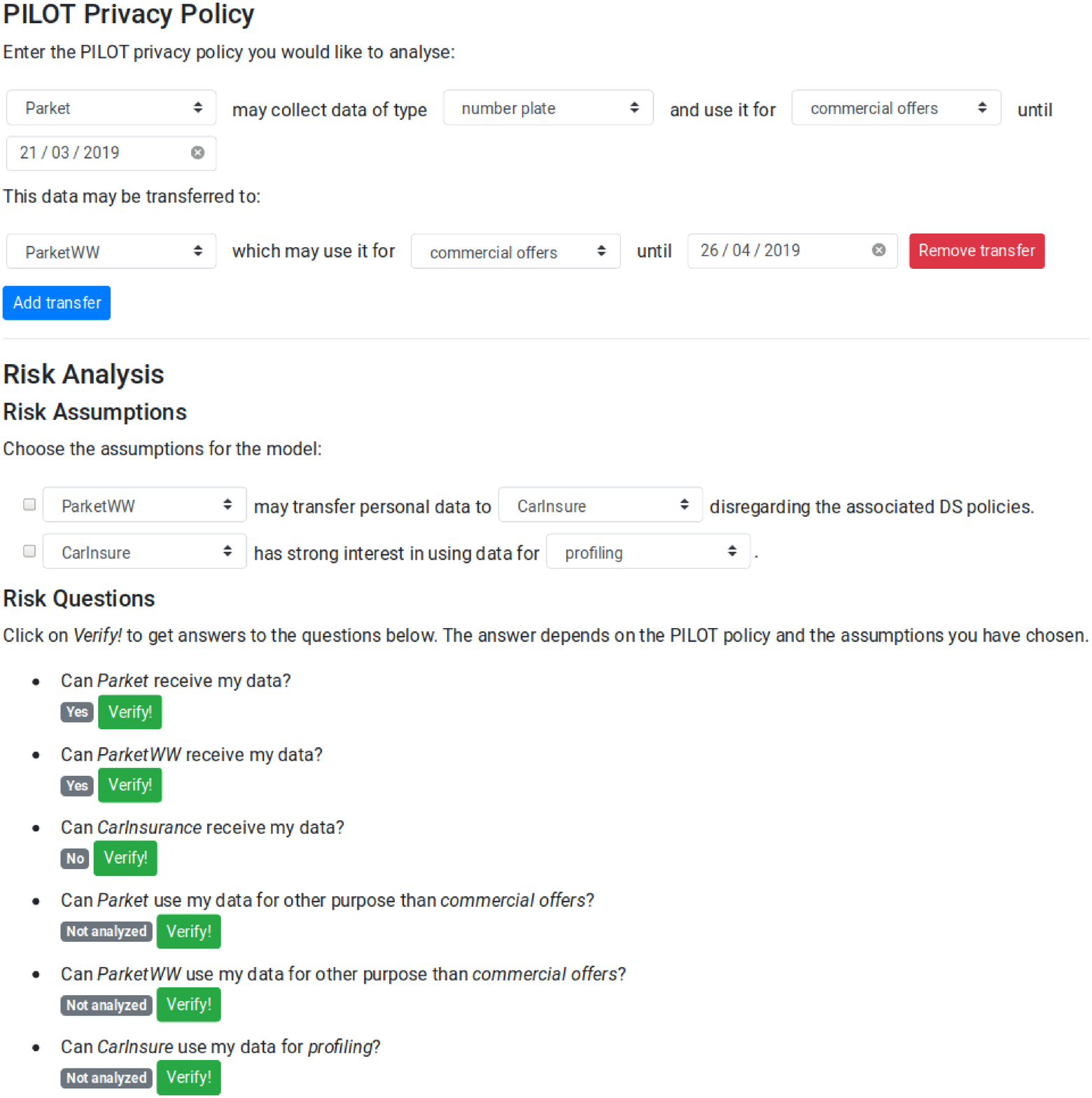}
\caption{Input Forms of Risk Analysis Web Application.}
\label{fig:web-app-risk-analysis}
\end{figure}

The web application is tailored to the ANPR case study we use throughout the report.
The \promela model and the policies defined in Section~\ref{subsec:case} are implemented in the application. This prototype can be generalized in different directions, for example by allowing users to enter specific risk assumptions on third parties. The range of questions could also be extended to include questions such as ``Can $X$ use $Y$'s data for other purpose than $\pur$? 
%
%
%
%
%
The code of the web application is available at~\cite{promelamodel}.

%% file: benefits-gdpr.tex
\section{Benefits of \pilot for the implementation of the GDPR}\label{sec:benefits-gdpr}

In this Section, we sketch the benefits of the use of \pilot in the context of the GDPR.
First, we believe that the adoption of a language like \pilot would contribute to reduce the imbalance of powers between \dcs and \dss without introducing prohibitive costs or unacceptable constraints for \dcs.
In addition, it can be used as a basis for a more effective consent mechanism as advocated by the WP29 in his opinion on the IoT.\footnote{``In practice, today, it seems that sensor devices are usually designed neither to provide information by themselves nor to provide a valid mechanism for getting the individual's consent. Yet, new ways of obtaining the user's valid consent should be considered by IoT stakeholders, including by implementing consent mechanisms through the devices themselves. Specific examples, like privacy proxies and sticky policies, are mentioned later in this document.''~\cite{wp29-eprivacy-2017}}
As demonstrated in \cite{lemetayer:hal-01953052}, \dss consents expressed through \pilot privacy policies can be produced automatically by a privacy agent implementing the abstract execution model sketched in Section 3. This agent can interact with the DS only in the cases  not foreseen by his privacy policy, for example a request from an unknown type of DC or for an unknown type of purpose. This makes it possible to reduce user fatigue while letting \dss in control of their choices.

In addition to the risk analysis described in the previous section, \pilot can be associated with verification tools to detect certain forms of non-compliance of DC privacy policies with respect to the GDPR. Examples of non-compliance include inconsistencies between the retention time and the purpose or between the purpose and the type of data. This would require the availability of a database of standard purposes and associated data types and retention times. In Europe, such a database could be provided, for example, by Data Protection Authorities or by the European Data Protection Board. It is also possible to detect privacy policies involving sensitive data\footnote{``Special categories of personal data'' as defined by Article 9 of the GDPR.} for which a stronger form of consent is required \cite{wp29-consent-2017}.

Since \pilot is defined through a precise execution model, the enforcement of  privacy policies can also be checked by a combination of means :
\begin{itemize}
  \item A priori (or ``static'')  verification of global properties such as ``No collection of data can take place by a DC if the DS has not previously received  the required information from this DC'' \cite{lemetayer:hal-01953052}. This particular property expresses a requirement of the GDPR regarding informed consent.
  \item On the fly (or ``dynamic'') verification of properties such as transfers to authorized third parties only. 
  \item A posteriori verification of properties in the context of audits. This type of verification can be implemented as an analysis of the execution traces (or logs) of the DC to support accountability.
  \end{itemize}
  
  It should be clear however, that we do not claim that all requirements of the GDPR regarding information and consent can be checked or even expressed using a privacy policy language like  \pilot.  For example, the general rules about imbalance of power or detriment stated by the WP29 \cite{wp29-consent-2017} are not prone to formal definition or automatic verification: they need to be assessed by human beings. Similarly, the fact that the \dcs authorized to receive the data and the allowed purpose are specific enough is a matter of appreciation rather than formal proof, even though an automatic verifier could rely on predefined databases of standard purposes and datatypes. The only claim that we make here is that most of what can be encapsulated into a machine-readable language and checked by a computer is included in \pilot. We compare more precisely \pilot with previous work in the next section.

%% file: related-work.tex
\section{Related Work}\label{sec:related-work}

Several languages or frameworks dedicated to privacy policies have been proposed.
A pioneer project in this area was the ``Platform for Privacy
Preferences'' (P3P)~\cite{RCppp99}.
P3P makes it possible to express notions such as purpose, retention time and
conditions.
However, P3P is not really well suited to the IoT as it was conceived as a policy language for websites.
Also, P3P does not offer support for defining data transfers.
Other languages close to P3P have been proposed, such as the ``Enterprise Policy Authorization Language'' (EPAL)~\cite{ashley2003enterprise} and ``An Accountability Policy Language'' (A-PPL)~\cite{appl}.
%
%
%
%
%
%
The lack of a precise execution model for these languages may also give rise to  ambiguities and variations in their implementations.

Recently, Gerl~\etal proposed LPL~\cite{LPLgbkb18} which is inspired
by GDPR requirements.
%
%
%
%
However, the lack of conditions and its centralized architecture makes LPL not suitable for IoT environments which are inherently distributed.
%
%
%
%

None of the above works include tools to help users understand the privacy risks associated with a given a policy, which is a major benefit of \pilot  as discussed in Section~\ref{sec:risk-analysis}.
In the same spirit, Joyee De~\etal~\cite{sjd-dlm} have proposed a methodology where \dss can visualize the privacy risks associated to their privacy settings.
Here the authors use harm trees to determine the risks associated with privacy settings.
The main difference with \pilot is that harm trees must be manually defined for a given application whereas we our analysis is fully automatic.\footnote{Only risk assumptions must be defined, which is useful to answer different ``what-if'' questions.}

Another line of work is that of formal privacy languages.
Languages such as S4P~\cite{S4Pbmb10} and SIMPL~\cite{Mfpmf08} define
unambiguously the behavior of the system---and, consequently,
the meaning of the policies---by means of trace semantics.
The goal of this formal semantics is to be able to prove global correctness properties such as ``\dcs always use \ds data according to their policies''.
While this semantics is well-suited for its intended purpose, it cannot be directly used to develop policy enforcement mechanisms. 
In contrast, we provide a \promela model in
Section~\ref{sec:risk-analysis}---capturing the execution model of
\pilot (cf. Section~\ref{sec:model})---that can be used as a reference to implement a system for the
enforcement for \pilot policies.
In addition, these languages, which were proposed before the adoption of the GDPR, were not conceived with its requirements in mind.

Other languages have been proposed to specify privacy regulations such as
HIPAA, COPAA and GLBA.
For instance, CI~\cite{BDMNpcifa06} is a dedicated linear temporal
logic based on the notion of contextual integrity.
CI has been used to model certain aspects of regulations such as HIPAA, COPPA and GLBA.
Similarly, PrivacyAPI~\cite{MGLpaactavlpp06} is an extension of the
access control matrix with operations such as notification and
logging.
The authors also use a \promela model of HIPAA to be able to
verify the ``correctness'' and better understand the regulation.
PrivacyLFP~\cite{DGJKDelshgpl10} uses first-order fixed point logic to
increase the expressiveness of previous approaches.
Using PrivacyLFP, the authors formalize HIPAA and GLBA with a higher
degree of coverage than previous approaches.
The main difference between \pilot and these languages is their focus.
\pilot is focused on modeling \dss and \dcs privacy policies and enhancing \dss awareness
whereas these languages  focus on modeling regulations.

Some access control languages such as XACML~\cite{xacml} and
RBAC~\cite{SCFYrbacm96} have been used for the specification of privacy
policies.
%
%
Typically, policies include the datatype to which they apply,
and a set of agents with privileges to perform certain actions---such as accessing the data.
Some extensions such as GeoXACML~\cite{geoaxcml} include
conditions depending on geolocation information.
However, none of these languages captures concepts such as
retention time, purpose or transfers.

Usage control (UCON)~\cite{PSuucm04,PHBduc06} appeared as an extension of access control to express how the data may be used after being accessed.
To this end, it introduces \emph{obligations}, which are actions such as ``do not transfer data item $i$''.
The Obligation Specification Language (OSL)~\cite{HPBSWplduc07} is an example of enforcement mechanism  through digital right management systems.
However, UCON does not offer any support to compare policies and does not differentiate between \dss and \dcs policies, which is a critical feature  in the context of privacy policies.  For \dss to provide an informed consent, they should know whether \dcs policies comply with their own policies. 
%

Some work has also been done on privacy risk analysis~\cite{sjd-dlm-book}, in particular to address the needs of the GDPR regarding Privacy Impact Assessments. We should emphasize that the notion of risk analysis used in this report is different in the sense that it applies to potential risks related to privacy policies rather than systems or products. Hence, the risk assumptions considered here concern only the motivation, reputation and potential history of misbehavior of the parties (but not the vulnerabilities of the systems, which are out of reach and expertise of the data subjects). 


%% file: conclusion.tex
\section{Conclusion}\label{sec:conclusion}

In this report, we have presented the privacy policy language \pilot, and a novel approach to analyzing privacy policies which is focused on enhancing informed consent.
An advantage of a language like \pilot is the possibility to use it as a basis to implement ``personal data managers'', to enforce privacy policies automatically, or ``personal data auditors'', to check a posteriori that a \dc has complied with the \ds policies associated with all the personal data that it has processed.
Another orthogonal challenge in the context of the IoT is to ensure that \dss are always informed about the data collection taking place in their environment and can effectively communicate their consent (or objection) to the surrounding sensors. Different solutions to this problem have been proposed in~\cite{lemetayer:hal-01953052} relying on \pilot as a privacy policy language used by \dcs to communicate their policies and \dss to provide their consent. These communications can either take place directly or indirectly (through registers in which privacy policies can be stored).

The work described in this report can be extended in several directions. First, the risk analysis model used here is simple and could be enriched in different ways, for example by taking into account risks of inferences between different types of data. The evaluation of these risks could be based  on past experience and research such as the study conducted by Privacy International.\footnote{\url{https://privacyinternational.org/sites/default/files/2018-04/data\%20points\%20used\%20in\%20tracking_0.pdf}} 
The risk analysis could also involve the history of the \ds (personal data already collected by \dcs in the past). On the formal side, our objective is to use a formal theorem prover to prove global properties of the model. This formal framework could also be used to implement tools to verify that a given enforcement system complies with the \pilot policies. 


%% file: appendix.tex
\section*{Appendix}

\input{subsumption}
\input{policy-join}

%% file: subsumption.tex
\section{Policy subsumption}
\label{sec:policy-subsumption}

We formalize the notion of \emph{policy subsumption}
as a relation over \pilot policies.
We start by defining subsumption of data usage and data communication
rules, which is used to define \pilot policy subsumption.

\begin{definition}[Data Usage Rule Subsumption]
  Given two data usage rules
  $\dur_1 = \langle P_1, \rt_1 \rangle$ and
  $\dur_2 = \langle P_2, \rt_2 \rangle$,
  we say that $\dur_1$ \emph{subsumes} $\dur_2$, denoted as
  $\dur_1 \dursubsumes \dur_2$, iff
  \begin{inparaenum}[i)]
    %
    %
    \item $\forall p_1 \in P_1 \cdot \exists p_2 \in P_2 \text{ such that } p_1 \popurposes p_2$; and
    \item $\rt_1 \leq \rt_2$.
  \end{inparaenum}
  \label{def:dur-subsumption}
\end{definition}

Intuitively, a data usage rule is more restrictive than another if:
\begin{inparaenum}[i)]
  \item the set of allowed purposes is smaller; and
  \item the data can be used for a shorter period of time.
\end{inparaenum}
For instance, $\langle \{\profiling\}, \dataprivacydaynum \rangle$
subsumes $\langle \{\profiling,\advertisement\},\gdprdaynum \rangle$
because \profiling is in both rules and
$\dataprivacydaynum$ $\leq \gdprdaynum$.


\begin{definition}[Data Communication Rule Subsumption]
  Given two data communication rules
  $\dcr_1 = \langle c_1, e_1, \dur_1 \rangle$ and
  $\dcr_2 = \langle c_2, e_2, \dur_2 \rangle$,
  we say that $\dcr_1$ \emph{subsumes} $\dcr_2$, denoted as
  $\dcr_1 \dcrsubsumes \dcr_2$, iff
  \begin{inparaenum}[i)]
    \item $c_1 \der c_2$;
    \item $e_1 \poentities e_2$; and
    \item $\dur_1 \dursubsumes \dur_2$.
  \end{inparaenum}
  \label{def:dcr-subsumption}
\end{definition}

A data communication rule is more restrictive than another if:
\begin{inparaenum}[i)]
  \item its conditions are stronger;
  \item the entity is more specific, i.e., less entities are included;
    and
  \item the usage rule is more restrictive.
\end{inparaenum}
For instance, consider the data communication rules for Parket in the
policies (\ref{eq:parket-policy}), denoted as $\dcr_1$, and
(\ref{eq:alice-policy}), denoted as $\dcr_2$, in Section~\ref{subsec:example}.
They have the same data usage rules and entity.
They only differ in the conditions.
The condition $\carlocation = \Lyon$ is clearly stronger than
$\true$ (denoted as $\true \der \carlocation = \Lyon$), since we
use $\true$ to impose no conditions (i.e., the rule is always active)
and $\carlocation = \Lyon$ to make the rule active when the
collecting camera is placed in Lyon.
Therefore, $\dcr_2$ subsumes $\dcr_1$---i.e.,
$\dcr_2 \dcrsubsumes \dcr_1$.

%
\begin{definition}[\pilot Privacy Policy Subsumption]
  Given two \pilot privacy policies
  $\pi_1 = \langle t_1, \dcr_1, \TR_1 \rangle$
  and
  $\pi_2 = \langle t_2, \dcr_2, \TR_2 \rangle$,
  we say that $\pi_1$ \emph{subsumes} $\pi_2$, denoted as
  $\pi_1 \subsumes \pi_2$ iff
  \begin{inparaenum}[i)]
    \item $t_1 \podata t_2$;
    \item $\dcr_1 \dcrsubsumes \dcr_2$; and
    \item $\forall \tr_1 \in \TR_1 \cdot \exists \tr_2 \in \TR_2
    \text{ such that } \tr_1 \dcrsubsumes \tr_2$.
  \end{inparaenum}
  \label{def:specification-policy-subsumption}
\end{definition}
A \pilot policy is more restrictive than another if:
\begin{inparaenum}[i)]
  \item the datatype is more specific;
  \item the data communication rule is more restrictive; and
  \item the set of allowed transfers is smaller.
\end{inparaenum}

As an example, consider the policies (\ref{eq:parket-policy}), denoted
as $p_1$, and (\ref{eq:alice-policy}), denoted as $p_2$, in
the example of Section~\ref{subsec:example}.
It is easy to see that $p_2$ subsumes $p_1$, i.e.,
$p_2 \subsumes p_1$, since they apply to the same datatype, the data
communication rule in $p_2$ is more restrictive than that of $p_1$ and
$p_2$ allows for no transfers whereas $p_1$ allows for transferring
data to ParketWW.


%

%
\section{Active Policies and Transfer rules}\label{sec:active-policy-transfer}
Here we formally define when \pilot policies and transfer rules are active.
Let \eval{\dataval}{d}{\phi} denote an \emph{evaluation function} for
conditions.
It takes as an input a formula ($\phi$) and it is parametrised by the
valuation function (\dataval) and device ($d$).
This function returns a boolean value,
$\{\texttt{true},\texttt{false}\}$, indicating whether a condition
holds, or the undefined value, $\bot$, when the information to
evaluate the condition is missing in the local state of the device.
\eval{\dataval}{d}{\phi} is defined as described in Table~\ref{tab:eval}.
We use a function $\timestamp(e) : \Events \to \nat$ to assign a
timestamp---represented as a natural number $\nat$---to each event of
a trace.

\begin{table}[t!]
  \centering
  \begin{tabular}{rcl}
    $\eval{\dataval}{d}{\true}$ & $=$ &  $\texttt{true}$\\
    $\eval{\dataval}{d}{\false}$ & $=$ & $\texttt{false}$\\
    $\eval{\dataval}{d}{i}$ & $=$ & $\dataval(d,i)$\\
    $\eval{\dataval}{d}{c}$ & $=$ & $\hat{c}$\\
    $\eval{\dataval}{d}{f(t_1,t_2,\ldots)}$ & $=$ & $\hat{f}(\eval{\dataval}{d}{t_1},\eval{\dataval}{d}{t_2},\ldots)$\\
    $\eval{\dataval}{d}{t_1 \ast t_2}$ & $=$ &
    $\begin{cases}
      \eval{\dataval}{d}{t_1} ~ \hat{\ast} ~ \eval{\dataval}{d}{t_2} \text{ if } \eval{\dataval}{d}{t_i} \not = \bot \\
      \bot \text{ otherwise}
    \end{cases}$\\
    $\eval{\dataval}{d}{\phi_1 \wedge \phi_2}$ & $=$ &
    $\begin{cases}
      \eval{\dataval}{d}{\phi_1} \text{ and } \eval{\dataval}{d}{\phi_2} \text{ if } \eval{\dataval}{d}{\phi_i} \not = \bot \\
      \bot \text{ otherwise}
    \end{cases}$\\
    $\eval{\dataval}{d}{\neg \phi}$ & $=$ &
    $\begin{cases}
      \text{not } \eval{\dataval}{d}{\phi} \text{ if } \eval{\dataval}{d}{\phi} \not = \bot \\
      \bot \text{ otherwise}
    \end{cases}$\\
  \end{tabular}
  \caption{Definition of \eval{\dataval}{d}{\phi}.
  We use $\hat{c}$, $\hat{f}$ and $\hat{\ast}$ to denote the interpretation of
  constants, functions and binary predicates, respectively.
  We assume that these interpretations are the same across all devices.
  }
  \label{tab:eval}
\end{table}

\begin{compactdesc}
\item[Active policy.]
\pilot policies, stored in a particular device, may be active depending on the
state of the system and the data item to be sent.
In order to determine whether a policy is active we use a boolean function
$\activePolicy(p,$ $\send(\sndr, \rcv, i),$ $ \st)$ which returns true if policy
$p$ is active when item $i$ is sent from device $\sndr$ to device $\rcv$ in
state $\st$.
%
%
Formally,
$$
\begin{array}{c}
  \activePolicy(p, \send(\sndr,\rcv,i), \st) = \\
          \type(i) \podata \datatype \wedge
          \eval{\dataval}{\sndr}{\phi} \wedge
          \timestamp(\st, \send(\sndr,\rcv,i)) < \rt ~ \wedge \\
          \devtoent(rcv) \poentities e
\end{array}
$$
\noindent
where $p = (\datatype,\langle \phi, e, \langle \_, \rt \rangle \rangle,\_)$
and $\st = \langle \dataval, \_, \_ \rangle$.
Intuitively, given $p = (\datatype,\langle \phi, e, \langle P, $ $\rt \rangle
\rangle, \TR)$, we check that:
\begin{inparaenum}[i)]
  \item the type of the data to be sent corresponds to the type of data
  the policy is defined for ($\type(i) \podata \datatype$);
  \item the condition of the policy evaluates to true
  ($\eval{\dataval}{\sndr}{\phi}$);
  \item the retention time for the receiver has not expired
  ($\timestamp(\send(\sndr,\rcv,i)) < \rt$); and
  %
  %
  \item the entity associated with the receiver device is allowed by the policy
  ($\devtoent(rcv) \poentities e$).
\end{inparaenum}

\item[Active transfer rule.]
Similarly, we use a boolean function $\activeTransfer($ $\tr,$ $p,$ $\transfer($
$\sndr,$ $\rcv,$ $i),$ $\st)$ which determines if transfer rule $\tr$ from
policy $p$ is active when item $i$ is transferred from device $\sndr$ to device
$\rcv$ in state $\st$.
In order for a transfer rule to be active, the above checks are performed on the
transfer rule $\tr$, and, additionally, it is required that the retention time
for the sender has not elapsed ($\timestamp($$\transfer($ $\sndr,$ $\rcv,$ $i))
< \rt$).
\end{compactdesc}

%% file: policy-join.tex
\section{Policy join}\label{sec:policy-join}

We present a \emph{join operator} for \pilot policies and prove that
the resulting policy is more restrictive than the policies used to
compute the join.
We first define join operators for data usage rules and data
communication rules, and use them to the join operator for \pilot
policies.
Let $\mini(e,e')
= 
\begin{cases} 
    e  & \text{if} ~ e \leq_{\mathcal{X}} e' \\
    e' & \text{otherwise}
\end{cases}$
be a function that, given two elements $e,e' \in \mathcal{X}$ returns
the minimum in the corresponding partial order $\leq_\mathcal{X}$.
Let $\Cap$ denote the intersection keeping the minimum of comparable
elements in the partial order of purposes.  Formally, given
$P,P' \in \Purposes$,
$ 
  P \Cap P' \triangleq (P \cap P') \cup P'' \text{ where } P'' = \{
  p \in P ~ | ~ \exists p' \in P' \text{ s.t. } p <
  p' \}
$.
%

\begin{definition}[Data Usage Rule Join]
  Given two data usage rules
  $\dur_1 = \langle P_1, \rt_1 \rangle$ and
  $\dur_2 = \langle P_2, \rt_2 \rangle$,
  the data usage rule join operator is defined as:
  $
  \dur_1 \durjoin \dur_2 = \langle
                               P_1 \Cap P_2,
                               \mini(\rt_1,\rt_2)
                           \rangle
  $.

\end{definition}


\begin{definition}[Data Communication Rule Join]
  Given two data communication rules
  $\dcr_1 = \langle c_1, e_1, \dur_1 \rangle$ and
  $\dcr_2 = \langle c_2, e_2, \dur_2 \rangle$,
  %
  the data communication rule join operator is defined as:
  $
  \dur_1 \dcrjoin \dur_2 = \langle
                               c_1 \wedge c_2,
                               \mini(e_1,e_2),
                               \dur_1 \durjoin \dur_2
                           \rangle
  $.
\end{definition}


\begin{definition}[\pilot Policy Join \policyjoin]\label{def:policy-join}
  Given two \pilot policies
  $p = (t_1, \dcr_1, \TR_1)$ and
  $q = (t_2, \dcr_2, \TR_2)$,
  the policy join operator is defined as:
  $
  \dur_1 \dcrjoin \dur_2 = (
                              \mini(t_1,t_2),
                              \dcr_1 \dcrjoin \dcr_2,
                              \{t ~ \dcrjoin ~ t' ~ | ~ t  \in \TR_1   ~ \wedge
                                                  ~ t' \in \TR_2   ~ \wedge
                                                  ~ t  \dcrsubsumes t' ~ \}
                           )
  $.
\end{definition}

\subsection{Privacy Preserving Join}

%
We say that an join operation is privacy preserving if the resulting
policy is more restrictive than both operands.
Formally,
\begin{definition}[Privacy Preserving Join]
  We say that $\policyjoin$ is \emph{privacy preserving} iff
  $
    \forall p,q \in \wfpp \cdot (p \policyjoin q) \subsumes p \wedge (p \policyjoin q) \subsumes q
  $.
\end{definition}

Intuitively, it means that it satisfies the preferences of both
policies.
If $p$ and $q$ correspond to the policies of \ds and \dc,
respectively, the resulting \pilot policy is more restrictive than
that of the \ds and \dc, thus:
\begin{inparaenum}[i)]
\item it is not allowed by the \dc to use the \ds data in any way not
  expressed in the \ds policy; and
\item the \dc will be able to enforce the policy---since it is more
  restrictive than the one she proposed, the policy will not contain
  anything that the \dc have not foreseen.
\end{inparaenum}

In what follows we prove that the operation $\policyjoin$ is privacy
preserving, Lemma~\ref{lem:joinpreserves}.
First, we prove two lemmas required for the proof of
Lemma~\ref{lem:joinpreserves}.

\begin{lemma}\label{lem:durjoinpreserves}
    Given two data usage rules $\dur_1, \dur_2 \in \wfur$ it holds that
    $ \dur_1 \durjoin \dur_2 \dursubsumes \dur_1 ~ \text{and} ~ \dur_1 \durjoin \dur_2 \dursubsumes \dur_2$.
    \vspace{-2mm}
\end{lemma}

\begin{proof}
    We split the proof into the two conjuncts of Lemma~\ref{lem:durjoinpreserves}.
    \begin{compactdesc}
        \item \underline{$\dur_1 \durjoin \dur_2 \dursubsumes \dur_1$} - We split the proof into the elements of data usage rules, i.e., purposes and retention time.
    
        \begin{itemize}
            \item We show that $\forall p \in \dur_1.P \Cap \dur_2.P \cdot \exists p' \in \dur_1.P ~ \text{such that} ~ p \popurposes p'$.
    
            \begin{itemize}
                \item $\forall p \in [(\dur_1.P \cap \dur_2.P) \cup \{p_1 \in \dur_1.P ~ | ~ \exists p_2 \in \dur_2.P \text{ s.t. } p_1 \popurposes p_2 \}] \cdot \exists p' \in \dur_1.P ~ \text{such that} ~ p \popurposes p'$ [By Def. $\Cap$]
                \item We split the proof for each operand in the union.
    
                \begin{itemize}
                    \item We show that $\forall p \in (\dur_1.P \cap \dur_2.P) \cdot \exists p' \in \dur_1.P ~ \text{such that} ~ p \popurposes p'$.
                    Assume $p \in (\dur_1.P \cap \dur_2.P)$.
                    Then $\exists p' \in \dur_1.P ~ \text{s.t.} ~ p=p'$ [By Def. $\cap$].
                    Therefore, $p \popurposes p'$.

                    \item We show that $\forall p \in \{p_1 \in \dur_1.P ~ | ~ \exists p_2 \in \dur_2.P \text{ s.t. } p_1 \popurposes p_2 \} \cdot \exists p' \in \dur_1.P ~ \text{such that} ~ p \popurposes p'$.
                    Assume $p \in \{p_1 \in \dur_1.P ~ | ~ \exists p_2 \in \dur_2.P \text{ s.t. } p_1 \popurposes p_2 \}$.
                    Then, $p \in \dur_1.P$, and, consequently, $\exists p' \in \dur_1.P ~ \text{s.t.} ~ p \popurposes p'$.
                \end{itemize}
            \end{itemize}
            \item $\mini(\dur_1.\rt,\dur_2.\rt) \leq \dur_1.\rt$ [By Def. \mini]
        \end{itemize}

        \item \underline{$\dur_1 \durjoin \dur_2 \dursubsumes \dur_2$} - Retention time is symmetric to the previous case, therefore we only show purposes.
        \begin{itemize}
            \item We show that $\forall p \in \dur_1.P \Cap \dur_2.P \cdot \exists p' \in \dur_2.P ~ \text{such that} ~ p \popurposes p'$.
    
            \begin{itemize}
                \item $\forall p \in [(\dur_1.P \cap \dur_2.P) \cup \{p_1 \in \dur_1.P ~ | ~ \exists p_2 \in \dur_2.P \text{ s.t. } p_1 \popurposes p_2 \}] \cdot \exists p' \in \dur_2.P ~ \text{such that} ~ p \popurposes p'$ [By Def. $\Cap$]
                \item We split the proof for each operand in the union.
    
                \begin{itemize}
                    \item The case $\forall p \in (\dur_1.P \cap \dur_2.P) \cdot \exists p' \in \dur_2.P ~ \text{such that} ~ p \popurposes p'$ is symmetric to the case above.
                    
                    \item We show that $\forall p \in \{p_1 \in \dur_1.P ~ | ~ \exists p_2 \in \dur_2.P \text{ s.t. } p_1 \popurposes p_2 \} \cdot \exists p' \in \dur_2.P ~ \text{such that} ~ p \popurposes p'$.
                    Assume $p \in \{p_1 \in \dur_1.P ~ | ~ \exists p_2 \in \dur_2.P \text{ s.t. } p_1 \popurposes p_2 \}$.
                    Then $\exists p_2 \in \dur_2.P$ s.t. $p \popurposes p_2$, and, consequently, $\exists p' \in \dur_2.P \text{ s.t. } p \popurposes p'$.
                \end{itemize}
        \end{itemize}
    \end{itemize}
    \end{compactdesc}
\end{proof}

\begin{lemma}\label{lem:dcrjoinpreserves}
    Given two data communication rules $\dcr_1, \dcr_2 \in \wfcr$ it holds that
    $ \dcr_1 \durjoin \dcr_2 \dursubsumes \dcr_1 ~ \text{and} ~ \dcr_1 \dcrjoin \dcr_2 \dcrsubsumes \dcr_2$.
    \vspace{-1mm}
\end{lemma}

\begin{proof}
    We split the proof into the two conjuncts of Lemma~\ref{lem:dcrjoinpreserves}.
    \begin{compactdesc}
        \item \underline{$\dcr_1 \durjoin \dcr_2 \dursubsumes \dcr_1$} - We split the proof into the elements of data communication rules, i.e., conditions and entities and data usage rules.
        \begin{itemize}
            \item $\dcr_1.c \wedge \dcr_2.c \der \dcr_1.c$. [By $\wedge$-elimination]
            \item $\mini(\dcr_1.e,\dcr_2.e) \poentities \dcr_1.e$ [By Def. \mini]
            \item $\dcr_1.\dur \durjoin \dcr_2.\dur \dcrsubsumes \dcr_1.\dur$. [By Lemma~\ref{lem:durjoinpreserves}]
        \end{itemize}
        \item \underline{$\dcr_1 \durjoin \dcr_2 \dursubsumes \dcr_2$} - The proof is symmetric to the previous case.
  \end{compactdesc}
\end{proof}

\begin{lemma}\label{lem:joinpreserves}
  The operation $\policyjoin$ in Definition~\ref{def:policy-join} is
  privacy preserving.
  \vspace{-1mm}
\end{lemma}

\begin{proof}
  Let $p_1$ and $p_2$ be two \pilot privacy policies. We show that 
  $ p_1 \policyjoin p_2 \subsumes p_1 ~ \text{and} ~ p_1 \policyjoin p_2 \subsumes p_2.$
  We proof each conjunct separately.
  
  \begin{compactdesc}
  \item \underline{$p_1 \policyjoin p_2 \subsumes p_1$} - We split the proof in cases based on the structure of \pilot policies, i.e., datatype ($p_1.t$,$p_2.t$), data communication rules ($p_1.\dcr, p_2.\dcr$) and transfers ($p_1.\TR, p_2.\TR$).
  
  \begin{itemize}
  \item $\mini(p_1.t,p_2.t) \podata p_1.t$. [By Def. of \mini]
  
  \item $p_1.\dcr \dcrjoin p_2.\dcr \dcrsubsumes p_1.\dcr$. [By Lemma~\ref{lem:dcrjoinpreserves}]
  
  \item $\forall \tr_\policyjoin \in \{\tr_1 \dcrjoin \tr_2 \mid \tr_1 \in p_1.\TR \wedge \tr_2 \in p_2.\TR \wedge \tr_1 \dcrsubsumes \tr_2\} \cdot \exists \tr \in \TR_1 ~ \text{s.t.} ~ \tr_\policyjoin \dcrsubsumes \tr$.
    Assume $\tr_\policyjoin \in \{\tr_1 \dcrjoin \tr_2 \mid \tr_1 \in p_1.\TR \wedge \tr_2 \in p_2.\TR \wedge \tr_1 \dcrsubsumes \tr_2\}$.
    Then $\tr_1 \in p_1.\TR$ and $\tr_\policyjoin \dcrsubsumes \tr_1$. [By Lemma~\ref{lem:dcrjoinpreserves}]
  \end{itemize}
  \item \underline{$p_1 \policyjoin p_2 \subsumes p_2$} - The proof is symmetric to the previous case.
  \end{compactdesc}
\end{proof}